\documentclass[12pt,english]{article}
\usepackage[letterpaper]{geometry}
\usepackage{amsmath}
\usepackage{amssymb}

\makeatletter
\usepackage{amsmath}
\usepackage{amsthm}
\usepackage{amssymb}

\theoremstyle{plain}
\newtheorem{theorem}{Theorem}[section]

\theoremstyle{plain}
\newtheorem{theorem-seq}{Theorem}

\newcommand\tnote[1]{}
\newcommand\inote[1]{}

\newcommand{\eqdef}{\stackrel{\rm{def}}{=}}
\DeclareMathOperator{\poly}{poly}

\newcommand{\E}{\mathbb{E}}
\newcommand{\V}{\mathbb{V}}

\newcommand{\N}{\mathbb{N}}

\newcommand{\F}{\mathbb{F}}

\newcommand{\C}{\mathcal{C}}

\def\eps{\varepsilon}

\newcommand{\NP}{\mathcal{NP}}

 \theoremstyle{plain}
 \theoremstyle{plain}
 \newtheorem{corollary}[theorem]{Corollary}  
 \theoremstyle{definition}
 \newtheorem{definition}[theorem]{Definition}
 \theoremstyle{plain}
 \newtheorem{lemma}[theorem]{Lemma}  
 \theoremstyle{plain}
 \theoremstyle{definition}
 
 \theoremstyle{plain}
 \newtheorem{proposition-seq}[theorem-seq]{Proposition}
 \theoremstyle{plain}
 
 \theoremstyle{plain}
 \newtheorem*{lemma*}{Lemma} 
 \theoremstyle{definition}
 
 \theoremstyle{plain}
 \newtheorem*{theorem*}{Theorem}
 \theoremstyle{plain}
 \newtheorem{claim}[theorem]{Claim}

\def\ccc{{$c^3$}}

\makeatother\makeatother

\usepackage{babel}

\makeatother

\usepackage{babel}

\begin{document}

\title{Dense locally testable codes cannot have constant rate and distance}

\author{Irit Dinur%
\thanks{Weizmann Institute of Science, ISRAEL. Email: \texttt{irit.dinur@weizmann.ac.il}.
Research supported in part by the Israel Science Foundation and by
the Binational Science Foundation and by an ERC grant.%
} \and Tali Kaufman%
\thanks{Bar-Ilan University and the Weizmann Institute of Science, ISRAEL. Email: \texttt{kaufmant@mit.edu} }}
\maketitle

\begin{abstract}

A $q$-query locally testable code (LTC) is an error correcting code that can be tested by a randomized algorithm that reads at most $q$ symbols from the given word.
An important question is whether there exist LTCs that have the \ccc~property: {\underline c}onstant relative rate, {\underline c}onstant relative distance, and that can be tested with a {\underline c}onstant number of queries. Such codes are sometimes referred to as ``asymptotically good''.

We show that {\em dense} LTCs cannot be \ccc.
The {\em density} of a tester is roughly the average number of distinct local views in which a coordinate participates. An LTC is {\em dense} if it has a tester with density $\omega(1)$.

More precisely, we show that a $3$-query locally testable code with a tester of density $\omega(1)$ cannot be \ccc. Moreover, we show that a $q$-locally testable code ($q>3$) with a tester of density $\omega(1)n^{q-2}$ cannot be \ccc. Our results hold when the tester has the following two properties: \begin{itemize}
  \item (no weights:) Every $q$-tuple of queries occurs with the same probability.
  \item (`last-one-fixed':) In every `test' of the tester, the value to any $q-1$ of the symbols determines the value of the last symbol. (Linear codes have constraints of this type).
\end{itemize}
We also show that several natural ways to quantitatively improve our results would already resolve the general \ccc~question, i.e. also for non-dense LTCs.
\end{abstract}

\inote{
TEXT ABSTRACT

A q-query locally testable code (LTC) is an error correcting code that can be tested by a randomized algorithm that reads at most q symbols from the given word.
An important question is whether there exist LTCs that have the ccc-property: *c*onstant relative rate, *c*onstant relative distance, and that can be tested with a *c*onstant number of queries. Such codes are sometimes referred to as ``asymptotically good''.

We show that *dense* LTCs cannot be ccc.
The *density* of a tester is roughly the average number of distinct local views in which a coordinate participates. An LTC is *dense* if it has a tester with density >> 1.

More precisely, we show that a 3-query locally testable code with a tester of density >> 1 cannot be ccc. Moreover, we show that a q-query locally testable code (q>3) with a tester of density >> n^{q-2} cannot be ccc. Our results hold when the tester has the following two properties:
1) "no weights":  Every q-tuple of queries occurs with the same probability.
2) "last-one-fixed": In every `test' of the tester, the value to any q-1 of the symbols determines the value of the last symbol. (Linear codes have constraints of this type).

We also show that several natural ways to quantitatively improve our results would already resolve the general ccc-question, i.e. also for non-dense LTCs.}

\global\long\def\set#1{\left\{  #1\right\}  }
\global\long\def\eqdef{\stackrel{{\rm def}}{=}}
\global\long\def\NP{\mathbf{NP}}
\global\long\def\aand{\hbox{\, and\,}}
\global\long\def\N{\mathbb{N}}
\global\long\def\span{\mbox{span}}
\global\long\def\unsat{{\rm UNSAT}}
\def\sat{{\rm SAT}}
\global\long\def\l{{\rm left}}
\global\long\def\r{{\rm right}}
\global\long\def\E{\mathbb{E}}
\global\long\def\V{\mathbb{V}}
\global\long\def\ekw{\varepsilon_{{\rm DP}}}
\global\long\def\ep{\varepsilon_{P^{2}}}
\global\long\def\eaps{\varepsilon_{aP^{2}}}
\global\long\def\ekw{\varepsilon_{{\rm IKW}}}
\global\long\def\eps{\varepsilon}
\global\long\def\ap{\alpha_{P^{2}}}
\global\long\def\akw{\alpha_{\mbox{IKW}}}
\global\long\def\aaps{\alpha_{aP^{2}}}

\global\long\def\cp{c_{P^{2}}}
\global\long\def\caps{c_{aP^{2}}}
\global\long\def\ckw{c_{\mbox{IKW}}}
\global\long\def\haps{h_{aP^{2}}}
\global\long\def\apx#1{\stackrel{#1}{\approx}}
\global\long\def\napx#1{\stackrel{#1}{\not\approx}}
\global\long\def\card#1{\left|#1\right|}
\global\long\def\K{\mathcal{K}}
\global\long\def\F{\mathbb{F}}
\global\long\def\poly{{\rm poly}}
\global\long\def\DB{\mathcal{\mathcal{DB}}}
\global\long\def\DBs{\DB_{\Lambda,t,l}}
\global\long\def\DBf{\DB_{\F,t,q-1}}
\global\long\def\D{\mathcal{D}}
\def\e0{\eps_0}
\global\long\def\sett#1#2{\left\{  #1\left|\;\vphantom{#1#2}\right.#2\right\}  }
\newcommand\remove[1]{{}}
\def\bits{{\set{0,1}}}
\def\dist{{\rm dist}}
\def\d{d}
\def\sd{{d^{1/2}}}
\def\b{{\beta}}

\section{Introduction}
An error correcting code is a set $\C\subset \Sigma^n$. The rate of the code is $\log\card \C / n$ and its (relative) distance is the minimal Hamming distance between two different codewords $x,y\in \C$, divided by $n$.
We only consider codes with distance $\Omega(1)$.

A code is called
{\em locally testable} with $q$ queries if it has a {\em tester}, which is a randomized algorithm with oracle access to the received word $x$. The tester reads at most $q$
symbols from $x$ and based on this local view decides if $x \in \C$ or
not. It should accept codewords with probability one, and reject words
that are far (in Hamming distance) from the code with noticeable
probability. The tester has parameters $(\tau,\eps)$ if \[
\forall x\in \Sigma^n ,\;\; \dist(x,\C) \ge \tau \qquad \Longrightarrow \qquad \Pr[\hbox{Tester rejects }x] \ge \eps\]

Locally Testable Codes (henceforth, LTCs) are studied extensively in recent years. A priori, even the existence of LTCs is not trivial. The Hadamard code is a celebrated example of an LTC, yet it is highly ``inefficient'' in the sense of having very low rate ($\log n/n$).
Starting with the work of Goldreich and Sudan~\cite{GS}, several more efficient constructions of LTCs have been given.  The best known rate for LTCs is $1/\log^{O(1)} n$, and these codes have $3$-query testers~\cite{BS05,Dinur,OrMeir}. The failure to construct \ccc-LTCs leads to one of the main open questions in the area: are there LTCs that are \ccc, i.e. {\underline c}onstant rate {\underline c}onstant distance and testable with a {\underline c}onstant number of queries (such codes are sometimes called in the literature ``asymptotically good''). The case of two queries has been studied in \cite{BSGS}. However, the case of $q>3$ is much more interesting and still quite open.

\paragraph{Dense testers.}
In this work we make progress on a variant of the \ccc~question. We show that LTCs with so-called dense testers, cannot be \ccc.

The density of a tester is roughly the average number, per-coordinate, of distinct local views that involve that coordinate. More formally, every tester gives rise to a constraint-hypergraph $H=([n],E)$ whose vertices are the $n$ coordinates of the codeword, and whose hyperedges correspond to all possible local views of the tester. Each hyperedge $h\in E$ is also associated with a constraint, i.e. with a Boolean function $f_h:\Sigma^q\to\bits$ that determines whether the tester accepts or rejects on that local view. For a given string $x\in \Sigma^n$, we denote by $x_h$ the substring obtained by restricting $x$ to the coordinates in the hyperedge $h$. The value of $f_h(x_h)$ determines if the string $x$ falsifies the constraint or not.

\begin{definition}[The test-hypergraph, density]
Let $C\subseteq \Sigma^n$ be a code, let $q\in \mathbb{N}$ and $\eps>0$.

Let $H$ be a constraint hyper-graph with hyperedges of size at most $q$. $H$ is an $(\epsilon,\tau)$-test-hypergraph for $\C$ if
\begin{itemize}
\item For every $x\in \C$ and every $h\in E$, $f_h(x_h) = 1$
\item For every $x\in \Sigma^n$, \[
\dist(x,\C)\ge\tau\quad\Rightarrow\quad \Pr_{h\in E}[f_h(x_h)=0] \ge \epsilon
    \] where $\dist(x,y)$ denotes relative Hamming distance, i.e., the fraction of coordinates on which $x$ differs from $y$.
\end{itemize}

Finally, the {\em density} of $H$ is simply the average degree, $\card E / n$.
\end{definition}

The hypergraph is equivalent to a tester that selects one of the hyperedges uniformly at random. Observe that we disallow weights on the hyperedges. This will be discussed further below.

Goldreich and Sudan~\cite{GS} proved that every tester with density $\Omega(1)$ can be made into a ``sparse'' tester with density $O(1)$ by randomly eliminating each hyper-edge with suitable probability.
This means that a code can have both dense and sparse testers at the same time. Hence, we define a code to have {\em density} $\ge d$ if it has a tester with density $d$. In this work we show that the existence of certain dense testers restricts the rate of the code.

We say that an LTC is {\em sparse} if it has no tester whose density is $\omega(1)$. We do not know of any LTC that is sparse. Thus, our work here provides some explanation for the bounded rate that known LTCs achieve.

In fact, one wonders whether density is an inherent property of LTCs. The intuition for such a claim is that in order to be locally testable the code seems to require a certain redundancy among the local tests, a redundancy which might be translated into density. If one were to prove that every LTC is dense, then it would rule out, by combination with our work, the existence of \ccc-LTCs.

In support of this direction we point to the work of the second author with co-authors (Ben-Sasson et al~\cite{BGKSV}) where it is shown that every linear LTC (even with bounded rate) must have some non-trivial density. I.e. they show that no linear LTC can be tested only with tests that from a basis to the dual code. Namely some constant density is required in every tester of an LTC.

\subsection{Our results}
We bound the rate of LTCs with dense testers. We only consider testers whose constraints have the ``last-one-fixed'' (LOF) property, i.e.  that the value to any $q-1$ symbols determine the value of the last symbol. Note for instance that any linear constraint has this property.

We present different bounds for the case $q=3$ and the case $q>3$ where $q$ denotes the number of queries.

\begin{theorem}\label{thm:q3}
Let $C\subseteq\bits ^n$ be a $3$-query LTC with distance $\delta$, and let $H$ be an $(\delta/3,\eps)$-test-hypergraph with density $\d$ and LOF constraints. Then, the rate of $C$ is at most $O(1 /\sd)$.
\end{theorem}

For the case of $q>3$ queries we have the following result

\begin{theorem}\label{thm:q}
Let $C\subseteq\bits ^n$ be a $q$-query LTC with distance $\delta$, and let $H$ be an $(\delta/2,\eps)$-test-hypergraph with density $\Delta$, where $\Delta =d n^{q-2}$, and LOF constraints. Then, the rate of $C$ is at most $O(1/d)$.
\end{theorem}

\paragraph{Extensions.}
In this preliminary version we assume that the alphabet is Boolean, but the results easily extend to any finite alphabet $\Sigma$. It may also be possible to get rid of the ``last-one-fixed'' restriction on the constraints, but this remains to be worked out.

\paragraph{Improvements.}
We show that several natural ways of improving our results will already resolve the `bigger' question of ruling out \ccc-LTCs.

\begin{itemize}
\item In this work we only handle {\em non-weighted} testers, i.e., where the hyper-graph has no weights. In general a tester can put different weights on different hyperedges. This is sometimes natural when combining two or more "types" of tests each with certain probability. This limitation cannot be eliminated altogether, but may possibly be addressed via a more refined definition of density. See further discussion Section~\ref{sec:nonweighted}.
\item In Theorem~\ref{thm:q3} we prove that $\rho \le O(1/d^{0.5})$. We show that any improvement of the $0.5$ exponent (say to $0.5+\eps$) would again rule out the existence of \ccc-LTCs, see Lemma~\ref{lemma:expo}
\item In Theorem~\ref{thm:q} we bound the rate only when the density is very high, namely, $\omega(n^{q-2})$. We show, in Lemma~\ref{lemma:q}, that any bound for density $O(n^{q-3})$ would once more rule out the existence of \ccc-LTCs. It seems that our upper bound of $\omega(n^{q-2})$ can be made to meet the lower bound, possibly by arguments similar to those in the proof of Theorem~\ref{thm:q3}
\end{itemize}

\paragraph{Related work.}
In the course of writing our result we have learned that Eli Ben-Sasson and Michael Viderman have also been studying the connection between density and rate and have obtained related results, through seemingly different methods.

\tnote{add refs}
\inote{split general q theorem to two thms: one for testers and one for codes with many high degree vertices in a constraint hyper-graph}

\section{Moderately dense $3$-query LTCs cannot be \ccc}

In this section we prove Theorem~\ref{thm:q3} which we now recall:\\


\noindent{\bf Theorem~\ref{thm:q3}.~}{\em
Let $C\subseteq\bits ^n$ be a $3$-query LTC with distance $\delta$, and let $H$ be an $(\delta/3,\eps)$-test-hypergraph with density $\d$ and LOF constraints. Then, the rate of $C$ is at most $O(1 /\sd)$.}\\

In order to prove the main theorem, we consider the hypergraph $H = (V,E(H))$ whose vertices are the coordinates of the code, and whose hyper-edges correspond to the different tests of the tester. By assumption, $H$ has $\d n$ distinct hyper-edges. We describe an algorithm in Figure~\ref{fig:alg} for assigning values to coordinates of a codeword, and show that a codeword is determined using $k = O(\frac n{\sd})$ bits.

We need the following definition. For a partition $(A,B)$ of the vertices $V$ of $H$, we define the graph $G_B = (A, E)$ where \[E = \sett{\set{a_1,a_2} \subset A}{\exists b\in B,\; \set{a_1,a_2,b}\in E(H)}.\]
A single edge $\set{a_1,a_2}\in E(G_B)$ may have more than one ``preimage'', i.e., there may be two (or more) distinct vertices $b,b'\in B$ such that both hyper-edges $\set{a_1,a_2,b},\set{a_1,a_2,b'}$ are in $H$. For simplicity we consider the case where the constraints are linear\footnote{More generally, when the constraints are LOF the set of all such $b$'s can be partitioned into all those equal to $w_b$ and all those equal to $1-w_b$.} which implies that for every codeword $w\in C$:  $w_{b} = w_{b'}$. This is a source of some complication for our algorithm, which requires the following definition.
\begin{definition}
Two vertices $v,v'$ are {\em equivalent} if \[\forall w\in C,\qquad w_v = w_{v'}.\] Clearly this is an equivalence relation. A vertex has {\em multiplicity $m$} if there are exactly $m$ vertices in its equivalence class. The reader is invited to assume, at first read, that all multiplicities are $1$.

Denote by $V^*$ the set of vertices whose multiplicity is at most $\b \sd$ for $\b=\alpha/16$.
\end{definition}

\begin{figure}[h]
\centering \fbox{%
\begin{minipage}[c]{6in}%
 \centering
\begin{enumerate}
\item[0.] {\bf Init:} Let $\alpha = 3\eps/\delta$ and fix $\b = \alpha/16$. Let $B$ contain all vertices with multiplicity at least $\b \sd$. Let $F$ contain a representative from each of these multiplicity classes. Let $B$ also contain all "fixed" vertices (whose value is the same for all codewords).
\item\label{step:clean} {\bf Clean:} Repeat the following until $B$ remains fixed:
\begin{enumerate}
\item\label{step:cleana} Add to $B$ any vertex that occurs in a hyper-edge that has two endpoints in $B$.
\item\label{step:cleanb} Add to $B$ all vertices in a connected component of $G_B$ whose size is at least $\b \sd$, and add an arbitrary element in that connected component into $F$.
\item\label{step:cleanc}  Add to $B$ any vertex that has an equivalent vertex in $B$.
\end{enumerate}

\item\label{step:S}  {\bf $S$-step:} Each vertex outside $B$ tosses a biased coin and goes into $S$ with probability $1/\sd$.
\remove{    The coin tosses are not fully independent, rather, they are conditioned on sending at most one vertex into $S$ from each multiplicity class.

    More explicitly: for each multiplicity class the decision is independent, and within a multiplicity class of size $t$ with probability $1-t/\sd$ no vertex is selected, and with the remaining probability a random one of the $t$ elements are selected.
}
    Let $B\leftarrow B\cup S$ and set $F\leftarrow F \cup S$.

\item If there are at least two distinct $x,y\in C$ such that $x_B=y_B$ goto step~\ref{step:clean}, otherwise halt.
\end{enumerate}
\end{minipage}} \caption{\label{fig:alg}The Algorithm}
\end{figure}

The following lemma is easy.
\begin{lemma}\label{lemma:bound}
If the algorithm halted, the code has at most $2^{\card F}$ words.
\end{lemma}
\begin{proof}
This follows since at each step setting the values to vertices in $F$ already fully determines the value of all vertices in $B$ (in any valid codeword). Once the algorithm halts, the values of a codeword on $B$ determines the entire codeword. Thus, there can be at most $2^{\card F} $ codewords.
\end{proof}

Let $B_t$ denote the set $B$ at the end of the $t$-th Clean step.  In order to analyze the expected size of $F$ when the algorithm halts, we analyze the probability that vertices not yet in $B$ will go into $B$ on the next iteration. For a vertex $v$, this is determined by its neighborhood structure. Let \[ E_v = \sett{\set{u,u'}}{u,u'\in V^*,\hbox{ and }\set{u,u',v}\in E(H)} \]
be a set of edges. Denote by $A$ the vertices $v$ with large $\card{E_v}$,
\[ A = \sett{v }{\card{E_v} \ge \alpha \d}.\]

The following lemma says that if $v$ has sufficiently large $E_v$ then it is likely to enter $B$ in the next round:
\begin{lemma}\label{lemma:alg}
For $t\ge 2$, if $v\in A$  then \[ \Pr_{S}[v\in B_t] \ge \frac 1 2.\]
\end{lemma}

Next, consider a vertex $v\not\in A$ that is adjacent, in the graph $G_{B_{t-1}}$, to a vertex $v'\in A$. This means that there is a hyper-edge $h=\set{v,v',b}$ where $b\in B_{t-1}$. If it so happens that $v'\in B_t$ (and the above lemma guarantees that this happens with probability $\ge \frac 1 2$), then the hyper-edge $h$ would cause $v$ to go into $B_t$ as well. In fact, one can easily see that if $v$ goes into $B_t$ then all of the vertices in its connected component in $G_{B_{t-1}}$ will go into $B_t$ as well (via step~\ref{step:cleana}). Let $A_t$ be the set of vertices outside $B_t$ that are in $A$ or are connected by a path in $G_{B_t}$ to some vertex in $A$. We have proved

\begin{corollary}\label{cor:alg}
For $t\ge 2$, let $v \in A_{t-1}$ then \[ \Pr_{S}[v\in B_t] \ge \frac 1 2. \]\qed
\end{corollary}

\begin{lemma}\label{lemma:Z}
If the algorithm hasn't halted before the $t$-th step and $\card{A_{t}}< \frac \delta 2 n$ then the algorithm will halt at the end of the $t$-th step.
\end{lemma}

Before proving the two lemmas, let us see how they imply the theorem.
\begin{proof}(of theorem)~
For each $t\ge 2$, Corollary~\ref{cor:alg} implies that for each $v\in A_t$ half of the $S$'s put it in $B_t$. We can ignore the sets $S$ whose size is above $2\cdot n/\sd$, as their fraction is negligible. By linearity of expectation, we expect at least half of $A_t$ to enter $B_t$.  In particular, fix some $S_{t-1}$ to be an $S$ that attains (or exceeds) the expectation. As long as $\card{A_t}\ge \delta n /2$ we get
\[
\card{B_t} \ge \card{B_{t-1}} + \card{A_t}/2 \ge  \card{B_{t-1}} + \delta n/4 .\]
Since $\card{B_t} \le n$ after $\ell \le 4/\delta$ iterations when the algorithm runs with $S_1,\ldots,S_\ell$ we must have $\card{A_\ell} < \delta n/2$.  This means that the conditions of Lemma~\ref{lemma:Z} hold, and the algorithm halts.

How large is the set $F$? In each $S$-step the set $F$ grew by $\card S \le 2 n/\sd$ (recall we neglected $S$'s that were larger than that). The total number of vertices that were added to $F$ in $S$-steps is thus $O(\ell \cdot n/\sd)$.

Other vertices are added into $F$ in the init step and in step~\ref{step:cleanb}. In both of these steps one vertex is added to $F$ for every $\b\sd$ vertices outside $B$ that are added into $B$. Since vertices never exit $B$, the total number of this type of $F$-vertices is $n/(\b\sd)$.

Altogether, with non-zero probability, the final set $F$ has size $O(\frac {1}{\sd})\cdot n$. Together with Lemma~\ref{lemma:bound} this gives the desired bound on the number of codewords and we are done.
\end{proof}
We now prove the two lemmas.
\subsection {Proof of Lemma~\ref{lemma:alg}}
We fix some $v\in A$. If $v\in B_{t-1}$ then we are done since $B_t\supseteq B_{t-1}$. So assume $v\not\in B_{t-1}$ and let us analyze the probability of $v$ entering $B_t$ over the random  choice of the set $S$ at iteration $t-1$. This is dictated by the graph structure induced by the edges of $E_v$. Let us call this graph $G=(U,E_v)$, where $U$ contains only the vertices that touch at least one edge of $E_v$. We do not know how many vertices participate in $U$, but we know that $\card{E_v} \ge \alpha d$.

We begin by observing that all of the neighbors of $u$ must be in the same multiplicity class\footnote{Or, more generally for LOF constraints, in one of a constant number of multiplicity classes.}. Indeed each of the edges $\set{v,u,u_i}$ is a hyper-edge in $H$ and the value of $u_i$ is determined by the values of $v$ and $u$. Therefore, the degree in $G$ of any vertex $u\in U$ is at most $\b\sd$, since vertices with higher multiplicity are not in $V^*$ and therefore do not participate in edges of $E_v$.

For each $u\in U$ let $I_u$ be an indicator variable that takes the value $1$ iff there is a neighbor of $u$ that goes into $S$. If this happens then either
\begin{itemize}
\item $u\in S$: this means that $v$ has a hyperedge whose two other endpoints are in $B_t$ and will itself go into $B_t$ (in step~\ref{step:cleana}).
\item $u\not\in S$: this means that the graph $G_{B_t}$ will have an edge $\set{v,u}$.
\end{itemize}
If the first case occurs for any $u\in U$ we are done, since $v$ goes into $B_t$ in step~\ref{step:cleana}. Otherwise, the random variable $\sum_{u\in U}{I_u}$ counts how many distinct edges $\set{v,u}$ will occur in $G_{B_t}$. If this number is above $\b\sd$ then $v$ will go into $B_t$ (in step~\ref{step:cleanb}) and we will again be done.
It is easy to compute the expected value of $I$. First, observe that
\[\E[I_u] = 1- (1-1/\sd)^{deg(u)}\]
where $deg(u)$ denotes the degree of $u$ in $G$ and since the degree of $u$ is at most $\b\sd$, this value is between $deg(u)/2\sd$ and $deg(u)/\sd$. By linearity of expectation \[\E[I] =
\sum_u \E[I_u] \ge \sum_u deg(u)/2\sd = \card{E_v}\d^{-1/2} \ge \alpha \sd. \]
We will show that $I$ has good probability of attaining a value near the expectation (and in particular at least $\alpha \sd/2 \ge \b\sd$), and this will put $v$ in $B_t$ at step \ref{step:cleanb}. The variables $I_u$ are not mutually independent, but we will be able to show sufficient concentration by bounding the variance of $I$, and applying Chebychev's inequality.

The random variables $I_{u}$ and $I_{u'}$ are dependent exactly when $u,u'$ have a common neighbor (the value of $I_u$ depends on whether the neighbors of $u$ go into $S$). We already know that having a common neighbor implies that $u,u'$ are in the same multiplicity class. Since $U \subset V^*$, this multiplicity class can have size at most $\b\sd$. This means that we can partition the vertices in $U$ according to their multiplicity class, such that $I_u$ and $I_{u'}$ are fully independent when $u,u'$ are from distinct multiplicity classes. Let $u_1,\ldots,u_t$ be representatives of the multiplicity classes, and let $d_i \le \b \sd$ denote the size of the $i$th multiplicity class. Also, write $u\sim u'$ if they are from the same multiplicity class.

\begin{eqnarray*}
Var[I] = \E[I^2 ]-(\E[I])^2 &=& \E\sum_{u,u'} I_u I_{u'} - \sum_{u,u'} \E I_u \E I_{u'} \\
&=& \sum_{u \sim u'} \E [I_{u}I_{u'}] + \sum_{u \not\sim u'} \E I_{u}\E I_{u'}   - \sum_{u,u'} \E I_u \E I_{u'} \\
&\le& \sum_i\sum_{u\sim u_i}\sum_{u'\sim u_i}  \E I_{u}I_{u'} \\
&\le& \sum_i\sum_{u\sim u_i}\sum_{u'\sim u_i} \E I_{u}\cdot 1 \\
&\le& \sum_i\sum_{u\sim u_i} \E I_u \cdot d_i
\le \sum_i\sum_{u\sim u_i} \frac{deg(u)}{\sd} \cdot \b \sd \\
&=& \b\sum_u {deg(u)}= 2\b \card{E_v}
\end{eqnarray*}
By Chebychev's inequality, \[ \Pr[ \card{I-\E [I]} \ge a] \le Var[I]/a^2 \]
Plugging in $a = \E[I]/2$ we get
\[ \Pr\left[ \card{I-\E [I]} \ge \frac{\E[I]}2\right] \le \frac{Var[I]}{(\E[I]/2)^2}  \le
(2\b\card{E_v})\cdot ({(\frac 1 2 \card{E_v}\d^{-1/2})^2 })^{-1} \le 8\b\d/\card{E_v} \le 8\b/\alpha.
 \]
and so by choosing $\beta = \alpha / 16$ this probability is at most a half. Thus, the probability that $I \ge \E I /2\ge \alpha \sd/2$ is at least a half. As we said before, whenever $I\ge \b \sd$ we are guaranteed that $v$ will enter $B_t$ in the next Clean step~\ref{step:cleanb} and we are done. \qed

\subsection{Proof of Lemma~\ref{lemma:Z}}

We shall prove that if the algorithm hasn't halted before the $t$-th step and $\card{A_{t}} < \frac \delta 2 n$ then $\card{B_t} > (1-\delta) n$. This immediately implies that the algorithm must halt because after fixing values to more than $1-\delta$ fraction of the coordinates of a codeword, there is a unique way to complete it.

Recall that $A$ is the set of all vertices $v$ for which $\card{E_v} \ge \alpha\d$. The set $B_t$ is the set $B$ in the algorithm after the $t$-th Clean step. The set $A_t$ is the set of vertices outside $B_t$ that are connected by a path in $G_{B_t}$ to some vertex in $A$.
Finally, denote $G = G_{B_t}$.

Assume for contradiction that $\card{B_t} \le (1-\delta) n$ and $\card{A_t}< \delta n/2$. This means that $Z = V\setminus(A_t\cup B_t)$ contains more than $\delta n/2$ vertices. Since $Z\cap A=\phi$, every vertex $v\in Z$ has $\card{E_v}<\alpha d$. Out contradiction will come by finding a vertex in $Z$ with large $E_v$.
If the algorithm doesn't yet halt, there must be two distinct codewords $x,y \in C$ that agree on $B_t$. Let $U_{x\neq y} = \sett{u\in V}{x_v\neq y_v}$. This is a set of size at least $\delta n$ tht is disjoint from $B_t$. Since $\card{A_t}\le \delta n /2$ there must be at least $\delta n /2$ vertices in $Z\cap U_{x\neq y}$. Suppose $u\in Z\cap U_{x\neq y}$ and suppose $u'$ is adjacent to $u$ in $G$. First, by definition of $Z$, $ u\in Z$ implies $u'\in Z$. Next, we claim that $u\in U_{x\neq y}$ implies $u'\in U_{x\neq y}$. Otherwise there would be an edge $\set{u,u',b}\in E(H)$ such that $b\in B_t$, and such that $x_u\neq y_u$ but both $x_{u'}=y_{u'}$ and $x_b=y_b$. This means that either $x$ or $y$ must violate this edge, contradicting the fact that all hyper-edges should accept a legal codeword. We conclude that the set $Z\cap U_{x\neq y}$ is a union of connected components of $G$. Since each connected component has size at most $\b\sd$ (otherwise it would have gone into !
 $B$ in a previous Clean step) we can find a set $D\subset Z\cap U_{x\neq y}$ of size $s$, for \[ \frac \delta 3 n \le \frac \delta 2 n-\b\sd \le s \le \frac \delta 2 n,\] that is a union of connected components, i.e. such that no $G$-edge crosses the cut between $D$ and $V\setminus D$.
Now define the hybrid word
\[ w = x_{D}y_{V\setminus{D}}\]
that equals $x$ on $D$ and $y$ outside $D$.
We claim that $dist(w,C) = dist(w,y) = \card D/n\ge \delta/3$. Otherwise there would be a word $z\in C$ whose distance to $w$ is strictly less than $\card D /n \le \delta /2$ which, by the triangle inequality, would mean it is less than $\delta n $ away from $y$ thereby contradicting the minimal distance $\delta n $ of the code.

Finally, we use the fact that $C$ is an LTC, \[ \dist(w,C)\ge \delta/3\qquad\Longrightarrow\qquad Prob_{h\sim E(H)}[h \hbox{ rejects }w] \ge \eps.\]
Clearly to reject $w$ a hyperedge must touch $D$. Furthermore, such a hyperedge cannot intersect $B$ on $2$ vertices because then the third non-$B_t$ vertex also belongs to $B_t$. It cannot intersect $B_t$ on $1$ vertex because this means that either the two other endpoints are both in $D$, which is imopssible since such a hyperedge would reject the legal codeword $x$ as well; or this hyperedge induces an edge in $G$ that crosses the cut between $D$ and $V\setminus D$. Thus, rejecting hyper-edges must not intersect $B_t$ at all.

Altogether we have $\eps d n $ rejecting hyperedges spanned on $V\setminus B_t$ such that each one intersects $D$. This means that there must be some vertex $v\in D$ that touches at least $\eps d n /(\delta n /3) = \alpha d$ rejecting hyperedges. Recall that $D\subset Z$ is disjoint from $A$, which  means that $\card{E_v}< \alpha d$. On the other hand, each rejecting hyperedge touching $v$ must add a distinct edge to $E_v$. Indeed recall that $E_v$ contains an edge $\set{u,u'}$ for each hyperedge $\set{u,u',v}$ such that $u,u'\in V^*$ and where $V^*$ is the set of vertices with multiplicity at most $\b\sd$. The claim follows since obviously all of the $\alpha d $ rejecting hyperedges are of this form (they do not contain a vertex of high multiplicity as these vertices are in $B$).\qed

\section{Very dense $q$-LTCs cannot be \ccc}
\def\d{\Delta}

In this section we prove the following theorem,\\

\noindent{\bf Theorem~\ref{thm:q}.~}{\em
Let $C\subseteq\bits ^n$ be a $q$-query LTC with distance $\delta$, and let $H$ be an $(\delta/2,\eps)$-test-hypergraph with density $\Delta$, where $\Delta =d n^{q-2}$, and LOF constraints. Then, the rate of $C$ is at most $O(1/d)$.}\\

Our proof is similar to the proof in the previous section. We describe an algorithm for assigning values to coordinates of a codeword, and show that a codeword is determined using $k \le n \cdot O(1/d)$ bits. As in the previous section, we use the following definitions. For a partition $(A,B)$ of the vertices $V$ of $H$, we define the $2$-graph $G_B = (A, E)$ where \[E = \sett{\set{a_1,a_2} \subset A}{\exists b_3, \cdots b_q \in B,\; \set{a_1,a_2,b_3, \cdots, b_q}\in E(H)}.\]


\begin{figure}[h]
\centering \fbox{%
\begin{minipage}[c]{6in}%
 \centering
\begin{enumerate}
\item[0.] {\bf Init:} Let $B=\emptyset$, $F=\emptyset$. Let $\alpha = \frac \eps{\delta/2}, \beta = \alpha / 6^q$.
\item\label{qstep:clean} {\bf Clean:} Repeat the following until $B$ remains fixed:
\begin{enumerate}
\item\label{qstep:cleana} Add to $B$ any vertex that occurs in a $q$-edge that has $q-1$ endpoints in $B$.
\item\label{qstep:cleanb} Add to $B$ all vertices in a connected component of $G_B$ whose size is at least $\beta d$, and add an arbitrary element in that connected component into $F$.
\end{enumerate}

\item\label{qstep:S}  {\bf $S$-step:} Each vertex outside $B$ tosses $q-2$ independent biased coins that get $1$ with probability $p=6^q/\alpha d$. A vertex goes into $S$ if it got $1$ in at least one of the $q-2$ coin tosses.
Let $B\leftarrow B\cup S$ and set $F\leftarrow F \cup S$.

\item If there are at least two distinct $x,y\in C$ such that $x_B=y_B$ goto step~\ref{qstep:clean}, otherwise halt.
\end{enumerate}
\end{minipage}} \caption{\label{fig:alg2}The Algorithm}
\end{figure}

%
%

The following two lemmas imply the theorem.
\begin{lemma}
If the algorithm halted, the code has at most $2^{\card F}$ words.
\end{lemma}
\begin{proof}
Identical to the case of $3$-queries.
\end{proof}

\begin{lemma}\label{lemma:q-alg}
Let $B_t$ denote the set $B$ at the end of the $t$-th Clean step. Let $v$ be a vertex whose $H$ degree is at least $\alpha\d$. Then if $v\not\in B_{t-1}$ the probability over the choice of $S$ that $v\in B_t$ is at least $\frac 1 2$.
\end{lemma}

Before proving the lemma, let us see how it implies the theorem.
\begin{proof}(of theorem)
 Let $L$ denote the vertices of degree less than $\alpha\Delta$. First, we prove that $\card L < \delta n/2$. Otherwise, $\card L \ge \delta n /2$ and let $L'\subset L$ be an arbitrary subset of $L$ of size $\delta n/2$.  Let $x\in C$ and consider the hybrid word $w$ defined to equal $x$ outside of $L'$ and $1- x$ on $L'$. Clearly \[ \dist(w,C) = \dist(w,x) = \delta /2\] since were there a closer word $x'\neq x$ to $w$ it would be less than $\delta$ away from $x$ by the triangle inequality. By the $(\delta/2,\eps)$-LTC property we know that $w$ is rejected with probability at least $\eps$, i.e., it is rejected by at least $\eps n \Delta$ hyperedges. But simple averaging shows there must be a vertex in $L'$ touching at least $\eps n \Delta / (\delta n/2) = \alpha \Delta $ hyperedges, contradicting the definition of $L'\subseteq L$.

 Denote by $B_t$ the set $B$ after the $t$-th Clean step. Also denote $A_t = V \setminus(B_t\cup L)$.
 Let $v\in A_t$, then by Lemma~\ref{lemma:q-alg} $v$ will enter $B_{t+1}$ with probability at least $1/2$. We expect, over the choice of $S$ that half of the vertices of $A_t$ will go into $B_{t+1}$, and thus \[ \E_S[\card{A_{t+1}}] \le \card{A_t}/2 \] Let $S^{(1)},\ldots,S^{(t)}$ be the sets that attain or exceed this expectation at steps $1,\ldots,t$ (again, wlog we ignore sets $S$ whose size deviates from their expected size which is at most $qn/d$). If the algorithm chooses these sets $S^{(1)}, S^{(2)}, \ldots$ then at the $t$-th step we have $\card{A_{t+1}} \le 1/2^t \cdot n$. For $t = \log 2/\delta$ this is no larger than $\delta n /2$. Since $L$ too is smaller than $\delta n /2$, we deduce that $\card{B_{t+1}} > (1-\delta)n$ and the algorithm must halt.

The size of the set $F$ when the algorithm halts is no more than $\log 2/\delta $ times twice the expected size of $S$ (which is at most $O(n/ d)$), plus no more than $n/(\beta d)$ (from the Clean steps). Altogether this is $O(n/d)$ and this bounds the rate by $O(1/d)$.
\end{proof}

Let us now prove Lemma~\ref{lemma:q-alg}.
\begin{proof}

Consider the set of ($q-1$)-edges \[ N_{q-1}(v) = \sett{\set{u_1, \cdots, u_{q-1}}}{\set{u_1, \cdots, u_{q-1},v} \in E(H)}.\]
While we know that $\card{N_{q-1}(v)} \ge \alpha\Delta = \alpha dn^{q-2}$, we do not know how many vertices participate in these edges.
Let us fix some arbitrary order converting each subset in $N_{q-1}(v)$ to an ordered tuple.

Each vertex $v$ outside $B$ tosses $q-2$ independent coins each has probability $p=\frac{6^q}{d}$ of getting $1$.
Let $S_i$, $1 \leq i \leq q-2$,  be the set of vertices that their $i$-th coin toss is $1$.
A vertex $v$ goes into $S$ if it gets $1$ in at least one of the $q-2$ independent coin tosses, i.e. $S$ is the union of all $S_i$'s.


For $1 \leq i \leq q-2$ we define $N_{i}(v)$ similar to the above. Namely
\[ N_{i}(v) = \sett{(u_1, \cdots, u_{i})}{(u_1, \cdots, u_{i},x) \in N_{i+1}(v) \mbox{ and } x \in S_{q-1-i}}.\]
We call the elements in $N_i(v)$ $i$-edges (even for $i=1$).

Our goal is to show that with probability greater than $\frac{1}{2}$ over the selection of $S$, the set $N_1(v)$ is of size greater than $\alpha \frac{d}{6^q}$. This would suffice to prove the lemma since this means that $v$ is in a large connected component and will go into $B$ in the next iteration.

An $i$-edge ${\set{u_1, \cdots, u_{i}}}$ is called $h$-{\em heavy} in $N_{i+1}(v)$ if the number of distinct
$x$'s for which $\set{u_1, \cdots, u_{i} ,x } \in N_{i+1}(v)$ is at least $h$.
For $1 \leq i \leq q-2$, let $H_{i}(v)$ be the set of $i$-edges that are $\frac{\alpha d}{2 \cdot 5^{q-2-i}}$ heavy in $N_{i+1}(v)$.

\begin{claim}\label{claim:many-heavy}
$|H_{i}(v)| \ge \frac{|N_{i+1}(v)|}{2 n}$ assuming $|N_{i+1}(v)| \geq \frac{ \alpha  dn^{i}}{5^{q-2-i}}$.
\end{claim}

\begin{proof}
Indeed, otherwise the number of $i+1$-edges in $N_{i+1}(v)$ is too low, namely, at most
$$ \mbox{number of heavy $i$-edges} \cdot n + \mbox{number of non-heavy $i$-edges} \cdot \frac{\alpha d}{2 \cdot 5^{q-2-i}}.$$
This is smaller than
$$\frac{|N_{i+1}(v)|}{2 n}\cdot n  + n^{i} \cdot \frac{\alpha d}{2 \cdot 5^{q-2-i}} = \frac{|N_{i+1}(v)|}{2} + \frac{\alpha dn^{i}}{2 \cdot 5^{q-2-i}} \leq |N_{i+1}(v)|.$$
\end{proof}

We next show that edges in $H_i(v)$ have very high probability of being selected into $N_{i}(v)$.


\begin{claim}\label{claim:high-prob-for-heavy-edges}
For $1 \leq i \leq q-2$, an edge in $H_{i}(v)$ goes into $N_{i}(v)$ with probability greater than
$p_i \eqdef  1-(1-p)^{\alpha d /2 \cdot 5^{q-2-i}} \geq 1 - \frac{1}{8q+1}$ over the selection of $S_1, \cdots, S_i$
\end{claim}

\begin{proof}
Consider ${\set{u_1, \cdots, u_{i}}} \in H_{i}(v)$. By the definition of $H_{i}(v)$
there are at least $\frac{\alpha d}{2 \cdot 5^{q-2-i}}$ distinct $x$'s such that
${\set{u_1, \cdots, u_{i},x}} \in N_{i+1}(v)$, ${\set{u_1, \cdots, u_{i}}}$ goes into $N_{i}(v)$ if at least one of these distinct $x$'s is selected into $S_{q-1-i}$. The probability that at least one is selected into $S_{q-1-i}$ is $p_i \eqdef  1-(1-p)^{\alpha d /2 \cdot 5^{q-2-i}}$.
Note that since $p = 6^q/\alpha d$, $p_i \geq 1 - \frac{1}{8q+1}$.
\end{proof}

%


We are now ready to show that for $1 \leq i \leq q-2$, with probability greater than $(1-\frac{1}{2q})^{i} > \frac{1}{2}$ over the selection of $S_1, \cdots, S_i$, $|N_{i}(v)|> \alpha dn^{i-1}/5^{q-1-i}$. Note that this implies that
$$|N_{1}(v)| \ge \alpha d/5^{q-2} > \alpha d/6^q.$$ This implies that $v$ is in a large connected component and hence will enter into $B$ in the next iteration.


\begin{claim}\label{claim:variance-bound}
For $1 \leq i \leq q-2$, let $N_i = |N_i(v)|$.
$$\Pr_{S_{q-1-i}} \left(N_i > \frac{1}{5n} N_{i+1} \left|  N_{i+1} > \frac{1}{(5n)^{q-2-i}}N_{q-1}\right.\right) > 1-\frac{1}{2q}$$
\end{claim}

\begin{proof}
For every $e \in H_i(v)$ we define an indicator random variable $I_e$ that gets $1$ iff $e$ is selected into $N_i(v)$, otherwise $I_e$ is $0$. Let $I= \sum_{e \in H_i(v)}I_e$. By Claim~\ref{claim:many-heavy} we have that if $N_{i+1} > \alpha d n^i / 5^{q-2-i} = \frac{1}{(5n)^{q-2-i}}N_{q-1}$ then $|H_i(v)| \geq \frac{N_{i+1}}{2n}$. Thus,
$$\E[N_i] \geq \E[I] = p_i|H_i(v)| .$$


The variance of $I$ can be bounded as follows. \inote{compute ad hasof}


\begin{eqnarray*}
Var[I]   & = & \sum_{e_1,e_2 \in H_{i}(v)}(E[I_{e_1} I_{e_2}] - E[I_{e_1}]E[I_{e_2}]) \leq  |H_{i}(v)|^2 p_i -  |H_{i}(v)|^2 p_i^2\\
         & = & |H_{i}(v)|^2(p_i-p_i^2) =\E^2[I] (\frac{1}{p_i} - 1) \leq \E^2[I] \frac{1}{8q }
\end{eqnarray*}

The last inequality holds since $p_i > 1 - \frac{1}{8q+1}$, which implies $\frac{1}{p_i} - 1 \leq \frac{1}{8q}$.

By Chebychev's inequality, \[ \Pr[ \card{I-\E [I]} \ge a] \le Var[I]/a^2 \]

Plugging in $a = \E[I]/2$ we get
\[ \Pr\left[ \card{I-\E [I]} \ge \frac{\E[I]}2\right] \le \frac{1}{2q}\]

Thus, the probability that $I \ge \E [I/2] \ge p_i |H_{i}(v)|/2 \geq (1 - \frac{1}{8q+1}) \cdot \frac{N_{i+1}}{2 \cdot 2 n} \ge \frac{N_{i+1}}{5n}$ is at least $(1-1/2q)$.
Thus, $$\Pr_{S_{q-1-i}} \left(N_i > \frac{1}{5n} N_{i+1} \left|  N_{i+1} > \frac{1}{(5n)^{q-2-i}}N_{q-1}\right.\right) > 1-\frac{1}{2q}.$$
\end{proof}

As a corollary of the last claim (Claim~\ref{claim:variance-bound}) we get the desired bound on $N_1(v)$:
\begin{corollary}
$$\Pr_{S_1,\cdots, S_{q-2}}(|N_{1}(v)|> \alpha d/5^q) > \frac{1}{2}$$
\end{corollary}

\begin{proof}
We prove by downwards induction on $i$ that
$$\Pr_{S_1, \cdots, S_{q-1-i}} (N_i > \frac{1}{(5n)^{q-1-i}}N_{q-1}) > (1-\frac{1}{2q})^{q-1-i}.$$

For $i=q-1$ this holds with probability $1$. By Claim~\ref{claim:variance-bound}, if the above holds for $i+1$ then it holds for $i$.
\end{proof}
The last corollary establishes the proof the the lemma.

%
%
%
%


\end{proof}

\section{Exploring possible improvements}\label{sec:counterex}

\subsection{Tradeoff between rate and density}\label{sec:tradeoff}

Any improvement over our bound of $\rho < 1/d^{1/2}$, say to a bound of the form $\rho < 1/d^{0.501}$ would already be strong enough to rule out \ccc-LTCs (with a non-weighted tester) regardless of their density.
The reason for this is the following reduction by Oded Goldreich.

Suppose, for the sake of contradiction, that there is some family
\begin{lemma}\label{lemma:expo}
Suppose for some $q\ge 3$ and some $\epsilon>0$ the following were true.\\

\begin{minipage}[h]{5in}
For any family $\set{C_n}$ of $q$-query LTCs with rate $\le \rho$ such that each $C_n$ has a tester with density at least $d$, then $\rho \le 1/d^{\frac 1{q-1}+\epsilon}$.
\end{minipage}\\

\noindent Then, there is no family of $q$-query LTCs with constant rate and any density, such that the tester is non-weighted.
\end{lemma}
\begin{proof}
Let $\beta = \frac 1 {q-1}+\epsilon$, and let $t\in \mathbb{N}$. Let $\set{C_i}$ be an infinite family of $q$-LTCs with density $d = O(1)$ and relative rate $\rho = \Omega(1)$. Then there is another infinite family $\set{\tilde C_i}$ of $q$-LTCs with density $d\cdot t^{q-1}$ and relative rate $\rho/t$. $\tilde C_i$ is constructed from $C_i$ by duplicating each coordinate $t$ times and replacing each test hyper-edge by $q^t$ hyperedges. Clearly the density and the rate are as claimed. The testability can also be shown. Plugging in the values $\tilde \rho = \rho/t$ and $\tilde d = d t^{q-1}$ into the assumption we get
\[ \rho/t = \tilde \rho \le 1/\tilde d^\beta = 1/(d t^{q-1})^\beta
\]
In other words $\rho d^\beta \le t ^{1-(q-1)\beta}$. Since $t$ is unbounded this can hold only if the exponent of $t$ is positive, i.e., $\beta \le 1/(q-1)$, a contradiction.
\end{proof}

\subsection{For $q>3$ density must be high}

\begin{lemma}\label{lemma:q}
Let $C$ be a $q$-LTC with rate $\rho$, and density $d$. Then there is a $(q+q')$-LTC $C'$ with density $d\cdot {n \choose {q'}}$ such that $C'$ has rate $\rho/2$, distance $\delta/2$.
\end{lemma}
\begin{corollary}
If there is a $3$-LTC with constant rate and density, then there are LTCs with $q>3$-queries, constant rate, and density $\Omega(n^{q-3})$.
\end{corollary}
The corollary shows that our upper bounds from Theorem~\ref{thm:q} are roughly correct in their dependence on $n$, but there is still a gap in the exponent.

\begin{proof}(of lemma)~
Imagine adding another $n$ coordinates to the code $C$ such that they are always all zero. Clearly the distance and the rate are as claimed. For the testability, we replace each $q$-hyper-edge $e$ of the hypergraph of $C$ with $n\choose {q'}$ new hyperedges that consist of the vertices of $e$ plus any $q'$ of the new vertices. The test associated with this hyperedge will accept iff the old test would have accepted, and the new vertices are assigned $0$. It is easy to see that the new hypergraph has average degree $d \cdot {n\choose q'}$. Testability can be shown as well.
\end{proof}

\subsection{Allowing weighted hypergraph-tests}\label{sec:nonweighted}
In this section we claim that when considering hyper-graph tests {\em with weights}, the density should not be defined as the ratio between the number of edges and the number of vertices. Perhaps a definition that takes the min-entropy of the graph into consideration would be better-suited, but this seems elusive, and we leave it for future work.

We next show that if one defines the density like before (ignoring the weights) then every LTC can be modified into one that has a maximally-dense tester. This implies that bounding the rate as a function of the density is the same as simply bounding the rate.

\begin{lemma}\label{lemma:nonweighted}
Let $C$ be a $q$-LTC with $q\ge 3$, rate $\rho$, distance $\delta$, and any density. Then there is another $q$-LTC $C'$ with a {\em weighted}-tester of maximal density $\Omega(n^{q-1})$ such that $C'$ has rate $\rho/2$, distance $\delta/2$.
\end{lemma}

\begin{corollary}\label{lemma:nonweighted}
Let $f:\mathbb{N}\to \mathbb{N}$ be any non-decreasing non-constant function. Any bound of the form $\rho \le 1/f(d)$ for weighted testers implies $\rho \le 1/f(n^{q-1})$, and in particular $\rho\to 0$.\qed
\end{corollary}

\begin{proof}(of lemma:)~
One can artificially increase the density of an LTC tester hypergraph $H$ by adding $n$ new coordinates to the code that are always zero, and adding all possible $q$-hyperedges over those coordinates (checking that the values are all-zero). All of the new hyper-edges will be normalized to have total weight one half, and the old hyperedges will also be re-normalized to have total weight one half. Clearly the rate and distance have been halved, and the testability is maintained (with a different rejection ratio). However, the number of hyperedges has increased to $n^q$ so the density is as claimed.
\end{proof}

\section*{Acknowledgement}
We would like to thank Oded Goldreich for very interesting discussions, and for pointing out the reduction in Section~\ref{sec:tradeoff}.

\end{document}